\documentclass[11pt,a4paper]{article}
\usepackage{graphicx}
\usepackage{color} 
\usepackage{slashed}
\usepackage[utf8]{inputenc}
\usepackage{amsmath}
\usepackage{amssymb}
\usepackage{amsthm}

\usepackage[pdfstartview=FitH,colorlinks=true,linkcolor=blue,anchorcolor=red,citecolor=magenta,urlcolor=blue]{hyperref}
\usepackage[english]{babel}
\usepackage{amsmath,amssymb,titling,authblk}
\usepackage{slashed}
\usepackage{amscd}
\usepackage[normalem]{ulem}
\usepackage{appendix}
\usepackage{bbold}

\usepackage{slashed}
\usepackage[top=2.3cm,right=2.3cm,left=2.3cm,bottom=2.3cm]{geometry}

\allowdisplaybreaks[4]
\numberwithin{equation}{section}

\definecolor{verde}{cmyk}{.83,.21,1,.08}
\definecolor{darkorchid}{rgb}{0.6, 0.2, 0.8}
\definecolor{darkgreen}{rgb}{0,.5,0}

\newtheorem{proposition}[equation]{Proposition}

\def\({\left(}
\def\){\right)}
\def\[{\left[}
\def\]{\right]}

\newcommand{\ii}{\mathrm{i}}

\newcommand{\dd}{\mathrm{d}}

\newcommand{\be}{\begin{equation}}
\newcommand{\ee}{\end{equation}}
\newcommand{\bea}{\begin{eqnarray}}
\newcommand{\eea}{\end{eqnarray}}
\def\nn{\nonumber}
\newcommand{\eqn}[1]{(\ref{#1})}
\newcommand{\del}{\partial}
\newcommand{\la}{\label}

\begin{document}
\title{Poisson gauge models and   Seiberg-Witten map}

\author[1]{V. G. Kupriyanov}
\author[2,3]{M. A. Kurkov}
\author[2,3]{P. Vitale}
\affil[ ]{}
\affil[1]{\textit{\footnotesize CMCC-Universidade Federal do ABC, 09210-580, Santo Andr\'e, SP, 
Brazil. }}
\affil[2]{\textit{\footnotesize Dipartimento di Fisica ``E. Pancini'', Universit\`a di Napoli Federico II, Complesso Universitario di Monte S. Angelo Edificio 6, via Cintia, 80126 Napoli, Italy.}}
\affil[3]{\textit{\footnotesize INFN-Sezione di Napoli, Complesso Universitario di Monte S. Angelo Edificio 6, via Cintia, 80126 Napoli, Italy.}}
\affil[ ]{}
\affil[ ]{\footnotesize e-mail: \texttt{vladislav.kupriyanov@gmail.com, max.kurkov@gmail.com, patrizia.vitale@na.infn.it}}
\maketitle

\abstract{The semiclassical limit of full non-commutative gauge theory is known as Poisson gauge theory. In this work we revise the construction of  Poisson gauge theory paying attention to the geometric meaning of the structures involved and advance in the direction of a further development of the proposed formalism, including the derivation of Noether identities and conservation of currents.    {For any linear  non-commutativity, $\Theta^{ab}(x)=f^{ab}_c\,x^c$, with $f^{ab}_c$ being  structure constants of a Lie algebra, an explicit form of the gauge  Lagrangian is proposed.} In particular a universal solution for the matrix $\rho$ defining the field strength and the covariant derivative is found.    {The previously known examples of $\kappa$-Minkowski, $\lambda$-Minkowski and rotationally invariant non-commutativity are recovered from the general formula.} The arbitrariness in the construction of Poisson gauge models  is addressed  in terms of Seiberg-Witten maps, i.e., invertible field redefinitions mapping gauge orbits onto  gauge orbits. 
}

\section{Introduction}
Recent investigations on  noncommutative field theory \cite{BBKL, {Kupriyanov:2020sgx}} have proposed a novel approach where gauge connections and field strengths are defined on the basis of two main requests. First, the commutative limit has to be  well defined and to give back  the standard Maxwell  theory. Second, the set of infinitesimal gauge transformations has to close a Lie algebra in such a way to be compatible with the underlying noncommutative spacetime, namely, a homomorphism has to exist between the latter and the noncommutative algebra of gauge parameters. These requests lead to a non-linear modification of Maxwell theory, solely dictated by compatibility of the noncommutative picture of space-time and theoretical consistency of electrodynamics. 

Interestingly, non-linearly modified Maxwell models, mainly based on phenomenological reasons,  go back to Born and Infeld in 1933 \cite{Born-Infeld} and a few years later by Euler and Heisenberg \cite{Euler-Heisenberg}. 
Since then, many  other models of
non-linear electrodynamics have been proposed and   analysed in  many areas of theoretical physics including  gravity, cosmology, string theory and condensed matter. A recent review  can be found in \cite{sorokin}. It would be  certainly interesting and worth to compare our findings with the latter, which in principle are based on  different premises and seem to be completely unrelated. We plan to come back to this issue in a future analysis. 

{ For the present paper, we shall deal with an $n$-dimensional manifold, $\mathcal{M}$  representing  space-time,  
 locally described by    $x^a$, $a = 0,\cdots, n-1$, a set of local coordinates. Non-commutative deformations of  space-time may be  characterised by the  Kontsevich   star product of functions on $\mathcal{M}$,  \cite{kontsevich}, which for each given Poisson bivector $\Theta^{ij}(x)$, reads
 \be
 f \star g = f\cdot g + \frac{\ii}{2}\{f,g\} + ...
 \ee
 where 
 \begin{equation}
\{f,g\}=   {\Theta^{ij}(x)}\,\partial_i f\,\partial_jg\, \la{pbr}
\end{equation}
 stands for the Poisson bracket associated with the Poisson tensor  $\Theta^{ij}(x)$, while the remaining terms, denoted through “...", contain higher derivatives of the functions $f$ and $g$. 
 
 For constant noncommutativity the deformation is 
  provided by the usual Moyal-Weyl star-product and its siblings \cite{moyal, voros} and the construction of
noncommutative gauge theories is standard and well-known (see \cite{szaborev, wess} for a review), although not completely satisfactory because of the known renormalizability issues. However, coordinate-dependent non-commutativity is non-trivial already at the classical level, because a differential calculus is needed, which ought to be compatible  with the associated $\star$ product. In the context of derivation based differential calculus the issue is discussed  in \cite{dbv, wallet, wess2}. For an approach with twisted differential calculus see for example \cite{twist1, twist2, twist3}. 

 The framework adopted in the present research is the one proposed in \cite{Kupriyanov:2020sgx}, where the noncommutative gauge theory is constructed by requiring compatibility of the gauge algebra with space-time noncommutativity, namely
 infinitesimal gauge transformations should close the non-commutative algebra
 \be
 \[\delta_f,\delta_g\] = \delta_{-\ii[f,g]_\star},\quad\quad f,g \in \mathcal{C}^{\infty}(\mathcal{M}), \la{fullgt}
 \ee 
with
\be
[f, g]_\star=   f\star g - g\star f  = \ii \{f,g\} + ...
\ee
  the star-commutator of the gauge parameters $f$ and $g$. Moreover,   the commutative limit is requested to  be well defined and standard.
 }

For a better understanding of the physical implications and  mathematical structures, we work in  a semiclassical approximation of spacetime non-commutativity. Therefore,  star commutators 
will be replaced by imaginary unit times the Poisson brackets.  In this approximation,  the full non-commutative algebra of gauge parameters defines the Poisson gauge algebra,
$
[\delta_f,\delta_g]=\delta_{\{f,g\}} ,
$ introduced in \cite{KS21}.
In the same work it was proposed an approach to the construction of (almost)-Poisson gauge transformations and the corresponding gauge algebra, based on the symplectic embedding of (almost)-Poisson structures. In a subsequent work \cite{Kup33}  Poisson gauge theory was further developed, as a dynamical field theoretical model with the Poisson gauge algebra representing its gauge symmetries. 

{
As we shall explain in detail in Sec.~\ref{Poiga}, the Poisson gauge theory is based on two essential ingredients. 
For a given Poisson bivector~$\Theta$  on $\mathcal{M}$ one has to construct:
\begin{itemize}
\item{the $n$ by $n$ matrix $\gamma(A,x)$, whose components $\gamma^{a}_{b}$ satisfy the first master equation
\be
\gamma_{ i }^{ j } \partial^{ i }_A \gamma^{ k}_{ l} - \gamma^{ k}_{ i } \partial_A^{ i } \gamma^{ j }_{ l} 
+ \Theta^{ j  i } \partial_{ i } \gamma^{ k}_{ l} 
- \Theta^{ k i } \partial_{ i } \gamma^{ j }_{ l} 
- \gamma^{ i }_{ l}\partial_{ i }\Theta^{ j  k} = 0,\la{master1}
\ee
\item{and the $n$ by $n$ matrix $\rho(A,x)$, which obeys the second master equation,
\be
\gamma^j_l \partial_A^l\rho_a^i + \rho_a^l \partial_A^i\gamma_l^j + \Theta^{jl}\partial_l\rho^i_a = 0. \la{master2}
\ee}}
\end{itemize} 
The former defines the deformed gauge transformations~\cite{Kupriyanov:2020sgx}, whilst the latter allows to introduce a covariant derivative and a  covariant field strength~\cite{Kup33}.  Hereafter we are using the notation $\partial_A^i  \equiv \frac{\partial}{\partial A_i(x)}$. In order to avoid confusions, we emphasise that in these equations the partial derivatives over coordinates act on the explicit dependence on $x$ only, whilst $A(x)$ is considered as an independent variable.

The solution of the first master equation has been constructed for an \emph{arbitrary} Poisson bivector $\Theta$, which is linear in $x$~\cite{KS21}, in terms of a single matrix-valued function. From now on we address this solution as the ``universal" one.   In the same work the deformed gauge transformation was expressed via symplectic geometric quantities, and the role of the matrix $\gamma$ in this geometric construction was clarified. However, the  role of the matrix $\rho$ in the  geometric formalism was unclear.  Moreover, the universal solution of the second master equation for the matrix $\rho$, in the same spirit as the one  for $\gamma$~\cite{KS21},   was also missing.

The present project fills both the gaps. First, we interpret  the deformed field strength in  the symplectic geometric entries and clarify the role of the matrix $\rho$ within this geometric construction. Second, we  construct the matrix $\rho$ for an \emph{arbitrary} Poisson bivector $\Theta$, which is linear in $x$, establishing that
 {\be
\rho^{-1} = \gamma - \ii\,\hat{A}.
\ee}
In this formula  $\gamma$ is given by the universal solution mentioned above, and $\hat{A}$ is the $n$  by $n$ matrix, constructed out of the gauge field $A$ and the structure constants $f^{ab}_{c}=\partial_c \Theta^{ab}$, whose components, by definition, read
\be
[\hat{A}]^b_c= -\ii f^{ab}_{c} A_a,   \la{irl}
\ee
see Sec. \ref{Poili} for details.
This universal solution for $\rho$, accompanied by the universal solution for $\gamma$, allows to construct the dynamical equations of motion for \emph{any}  Poisson bivector, linear in $x$. As we shall see in detail, there are other solutions of the master equations, i.e. for a given Poisson bivector $\Theta$ one can build different Poisson gauge theories. In the present article} we will discuss in more details the arbitrariness in the definition of the Poisson gauge theory and provide the corresponding Seiberg-Witten maps in some particular examples, {confronting the solutions, obtained in this paper with  the ones,  known before.}

{The paper is organized as follows. In sections \ref{Poiga} - \ref{Maxpoi} we review in detail the formulation of Poisson gauge theory by clarifying its geometric content and discussing Noether identities and conserved currents. {The main novelties of Sec.~\ref{Poiga} are the expression for the deformed field strength and gauge covariant derivative in terms of symplectic geometric structures, together with   the role of the matrix $\rho$ within the  geometric construction.}
Starting from section \ref{Poili} we concentrate on linear noncommutative structures. 
In particular we obtain a universal solution for the matrix $\rho$, Eq.~\eqref{irl}, which enters the definition of the covariant derivatives and the field strength. In Sections \ref{cases} and \ref{kappa} we illustrate our general findings in two noteworthy cases,    namely the $\mathfrak{su}(2)$-like non-commutativity and the $\kappa$-Minkowski space.  {For the latter 
we obtain new solutions 
 for the matrices $\gamma$ and $\rho$, 
 different from those previously obtained in \cite{KKV}. The reason for that is non-uniqueness of the solution of corresponding master equations.}  Section \ref{seiwit} discusses this arbitrariness in terms of Seiberg-Witten maps and exhibits explicit maps for the linear cases considered. We conclude with a discussion Section and an appendix. }

\section{Poisson gauge theory}\label{Poiga} 

In this section we will describe the main points of the construction of the Poisson gauge theory based on the symplectic embedding of Poisson manifolds.  The main emphasis will be on the geometric meaning of all appearing objects. {First we review the known results, regarding the gauge transformations and the matrix $\gamma$. After that we discuss the deformed field strength and the deformed covariant derivatives. In particular, we present  new expressions of these objects in terms of geometric structures  and we outline  the role of the matrix $\rho$ in the symplectic geometric construction. }
\subsection{{Gauge transformations.}}
The Poisson gauge transformations $\delta_fA$ should satisfy the following two conditions: they should close  the gauge algebra,
\begin{equation}
[\delta_f,\delta_g]A=\delta_{\{f,g\}}A\,,\label{ga}
\end{equation}
and reproduce the standard $U(1)$ gauge transformations in the classical limit, 
\be
\lim_{ {\Theta} \to 0}\delta_fA_a=\partial_af. \la{gcomlim}
\ee
If $\Theta^{ij}$ is constant, one may easily see that the expression, 
\begin{equation}
\delta^{can}_fA=d f+ \{A,f\}_{can}\,,\label{gtcan}
\end{equation}
 satisfies (\ref{ga}). However, for non-constant $\Theta^{ij}$ the standard Leibniz rule with respect to the partial derivative is violated, $\partial_a\{f,g\}\neq\{\partial_af,g\}+\{f,\partial_ag\}$, therefore the same expression will not close the algebra (\ref{ga}) anymore, it being 
 \be\label{gal}
 [\delta_f,\delta_g]A=  d\{f,g\} + \{A,\{f,g\}\} - d\Theta^{ij}(x)\,\del_if\del_j g.
 \ee
  To overcome this difficulty one has to modify the expression for the gauge transformations (\ref{gtcan}) introducing  corrections which compensate the last unwanted term in \eqn{gal}.
  {A suitable modification has the following form~\cite{Kupriyanov:2020sgx}:
  \begin{equation}\label{gtAprimo}
\delta_f A_a =\gamma^r_a(A)\,\partial_r f(x) +\{A_a(x),f(x)\},
\end{equation}
where the matrix $\gamma$ satisfies the first master equation~\eqref{master1}. Such a construction has an elegant geometric interpretation in terms of symplectic embeddings and constraints, described below.} The idea of symplectic embeddings of Poisson manifolds is quite general and widely used in many different contexts. It may be traced back to   the so called symplectic realizations introduced by  Weinstein \cite{wein} and further developed in \cite{crainic}. The approach followed here is due to \cite{Kup21}.

{Both the gauge transformation~\eqref{gtAprimo} and the first master equation~\eqref{master1} can be obtained in a simple way, considering } an extended space,  by means of symplectic embedding techniques \cite{KS21}. 
{The formalism is essentially based on  {two} ingredients. The first key ingredient of the construction is the symplectic embedding itself:}  to each coordinate $x^i$ of the initial space-time ${\mathcal{M}}$ one associates  a conjugate variable $p_i$, in such a way that the corresponding extended Poisson brackets on $T^*\mathcal{M}$,
\begin{equation}\label{PB1}
\{x^i,x^j\}=\Theta^{ij}(x)\,,\qquad\{x^i,p_j\}=\gamma^i_j(x,p)\,,\qquad \{p_i,p_j\}=0\,,
\end{equation}
 satisfy the Jacobi identity, under the condition that the matrix $\gamma^i_j(x,p)$  be non-degenerate.   For constant $\Theta^{ij}$ one finds $\gamma^i_j(x,p)=\delta^i_j$, so that $\{f(x),p_i\}=\partial_if(x)$, namely  the Poisson bracket with the auxiliary variable $p_i$  is just a partial derivative of $f$. For $\Theta^{ij}(x)$   not constant the expression for $\gamma^i_j(x,p)$ is more complicated. The Jacobi identity for the algebra (\ref{PB1}) implies the partial differential equation  \cite{Kup21},
\begin{equation}
\gamma^l_b\,\partial^b_p\gamma_a^k-\gamma^k_b\,\partial^b_p\gamma^l_a+\Theta^{lm}\,\partial_m\gamma_a^k-\Theta^{km}\,\partial_m\gamma_a^l-\gamma^m_a\,\partial_m\Theta^{lk}=0\,\label{first}
\end{equation}
with $\partial_m=\partial/\partial x^m$ and $\partial^b_p=\partial/\partial p_b$. 
{The second key ingredient of the construction is the set of  constraints 
\begin{equation}
\Phi_a:=p_a-A_a(x)\label{Phi}, \quad\quad a =1,...,n
\end{equation}
which allows to get rid of the auxiliary variables $p$. Imposing this constraint on~\eqref{first} we get exactly Eq.~\eqref{master1}. In other words, \emph{the first master equation is simply the Jacobi identity in the extended space, obtained via the symplectic embedding. }}

In order to understand how this construction allows for the generalization~\eqref{gtAprimo} of gauge transformations in presence of a non-trivial Poisson bracket in space-time, let us review the procedure, following  \cite{Kurkov:2021kxa}. We first consider the standard setting of $U(1)$ gauge theory, with $\Theta=0$. Then, the cotangent bundle $T^*\mathcal{M}$ is endowed with the canonical symplectic form $\omega_0$ which locally reads  $dp_i\wedge d x^i$.  The gauge field $A$ is a local one-form on $\mathcal{M}$. It is therefore associated with a local section of the cotangent space $T^*\mathcal{M},\,  s_A: \mathcal{U}\rightarrow T^*\mathcal{U}$, through a local trivialisation, $\psi_\mathcal{U}^{-1}(s_A(x))= (x, A(x))$, where $\mathcal{U}$ is a local chart on $\mathcal{M}$.   
Let 
\be\label{xiA}
\xi_A= \lambda_0-\pi^*A
\ee
be a  one-form on $T^*\mathcal{U}$ with $\lambda_0$ the Liouville form (locally equal to $p_i dx^i$) and $\pi:T^*\mathcal{M}\rightarrow \mathcal{M}$ the usual projection map. The one-form $\xi_A$ vanishes locally, through the pullback $s_A^*$
\be s_A^*(\xi_A)=0\label{constraint}
\ee
because $s_A^*(\lambda_0)=A= (\pi\circ s_A)^*(A)$. This means that $\xi_A$ vanishes exactly on ${\rm im}(s_A)\subset T^*\mathcal{U}$. Therefore the latter is identified by the constraint \eqn{constraint}. This amounts to   fix $p$, which is the fibre coordinate at $x$,  to its value $A(x)$ identified by the section $s_A$, i.e. {the relation~\eqref{constraint} corresponds to  the constraint~\eqref{Phi}, rewritten in the language of differential geometry. }
Then, the infinitesimal gauge transformation of the gauge potential $A$, with gauge parameter $f$, may be defined in terms of the canonical Poisson bracket $\omega_0^{-1}$ as follows
\be\label{deltacan}
\delta_f A_n (x) = s_A^*\{\pi^*f, \xi_{A_n}\}_{\omega_0^{-1}}= \frac{\del f}{\del x^m }\frac{\del  \Phi_{n}}{\del p_m}= \del_n f\,,
\ee
that is to say, loosely speaking, one first  performs the Poisson bracket in $T^*\mathcal{U}$, then goes to its local form through $s_A$, recovering the standard infinitesimal gauge variation of the potential. 

The  complicated procedure described above, certainly redundant in the canonical case $\Theta=0$, is  extremely useful and constructive 
for the case $\Theta\ne 0$. One first performs a symplectic embedding  of $(\mathcal{M}, \Theta)$ as described above, with symplectic one-form $\omega$ given by the inverse of the non-degenerate Poisson bracket  \eqn{PB1}, with Poisson tensor
\be\label{poitens}
   {\Pi(x,p)}= \Theta^{ij}(x)\del_i\wedge \del_j  +\gamma^{i}_j (x,p)\del_i\wedge  \del^j_p
\ee
while  the image of $U\subset \mathcal{M}$ through the local section $s_A$ is still defined by the constraint \eqn{xiA}. Then, the infinitesimal gauge transformation of the gauge potential is formally defined in the same way  as in the previous case,  \eqn{deltacan}, except for the fact that the canonical Poisson bracket is to be replaced by the Poisson bracket  \eqn{PB1}.  Therefore we find
\bea
\delta_f A_a:&=& s_A^*\{\pi^*f, \xi_{A_a}\}_{\Pi}= \Theta^{rs} \frac{\del f}{\del x^r}\frac{\del \xi_{A_a}}{\del x^s}+\gamma^r_s(x,A) \frac{\del f}{\del x^r} \frac{\del \xi_{A_a}} {\del p_s}\nn\\
&=&\{A_a, f\}_\Theta+ \gamma^r_a(x,A) \frac{\del f}{\del x^r},
\eea
{what is nothing but Eq.~\eqref{gtAprimo}.  }
The result can be restated in a simpler language, by replacing the constraint~\eqref{constraint} with its local form~\eqref{Phi}
\begin{equation}\label{gtA}
\delta_f A_a:=\{f,\Phi_a\}_{\Phi_a=0}=\{A_a(x),f(x)\}+\gamma^r_a(A)\,\partial_r f(x)\,.
\end{equation}
{Summarising, we see that\emph{ the  nontrivial gauge transformation~\eqref{gtAprimo} simply corresponds to the Poisson bracket~\eqref{gtA} of the gauge parameter with the constraint~\eqref{Phi} in the extended space, obtained via a symplectic embedding. The role of the matrix $\gamma$ is also clear: it defines the symplectic embedding~\eqref{PB1}. } }

\subsection{{Field strength and covariant derivative}}
Now let us proceed to the definition of the field strength. For  constant Poisson bracket $\Theta$ one may simply set,
\begin{equation}
F^{can}_{ab}:=\{p_a-A_a(x),p_b-A_b(x)\}=\partial_aA_b-\partial_bA_a+\{A_a,A_b\}_{can}\,.
\end{equation}
This quantity transforms covariantly under the gauge transformation (\ref{gtcan}), namely, $\delta_f^{can}F^{can}_{ab}=\{F^{can}_{ab},f\}_{can}$, and reproduces the standard $U(1)$ field strength in the commutative limit,  with $\lim_{ {\Theta}\to0} F^{can}_{ab} = \partial_aA_b-\partial_bA_a$. So, following the logic of the symplectic embedding it would be reasonable to test the same structure for coordinate-dependent Poisson brackets,  by posing
\begin{equation}
F_{ab}:=\{\Phi_a,\Phi_b\}_{\Phi=0}=\gamma_a^l(A)\,\partial_lA_b-\gamma_b^l(A)\,\partial_lA_a+\{A_a(x)\,A_b(x)\}\,.
\end{equation}
However, by checking its behaviour under Poisson gauge transformation  (\ref{gtA}) we find (see details in Appendix \ref{AppB})
\begin{equation}
\delta_f F_{ab}=\{F_{ab},f\}+\left(\partial^c_p\{f,p_a\}\right)_{\Phi=0}F_{cb}-\left(\partial^c_p\{f,p_b\}\right)_{\Phi=0}F_{ca}\,.\label{gccF}
\end{equation}

The first term  is exactly what we need for the gauge covariance condition, but the other two terms are undesirable. 

Interestingly, a solution may be found by performing a  transformation in the basis of constraints 
\begin{equation}
\Phi_a\,\,\to\,\,\Phi^\prime_a:=\rho_a^m(A)\,\Phi_m\,,\label{Phi1}
\end{equation}
with  $\rho_a^m(A)$ a non-degenerate matrix to be determined by the covariance request. The non-degeneracy ensures that  
\begin{equation}
\Phi^\prime_a=0\,,\qquad\Leftrightarrow\qquad \Phi_a=0\,.
\end{equation}
The field strength is thus defined by means of  the new basis of constraints according to\footnote{Notice that the field strength was already defined in previous works as a deformation of the classical one; the two definitions yield exactly the same expression, the advantage being here in the geometric interpretation that appears to be   more natural.}
\begin{equation}
{\cal F}_{ab}:=\{\Phi^\prime_a,\Phi^\prime_b\}_{\Phi^\prime=0}=\rho_a^m(A)\,\rho_b^n(A)\,F_{mn}\,.\label{pfs}
\end{equation}
By computing  its  gauge variation  we find
\begin{eqnarray}
\delta_f {\cal F}_{ab}=\delta_f \left(\rho_a^m(A)\right)\rho_b^n(A)\,F_{mn}+\rho_a^m(A)\delta_f\left(\rho_b^n(A)\right)F_{mn}+\rho_a^m(A)\,\rho_b^n(A)\,\delta_fF_{mn}\,.\label{gcc1}
\end{eqnarray}
Observing that,
\begin{equation}\label{w7a}
\delta_f \rho_a^m(A)=\partial_A^b\rho_a^m(A)\,\{f,\Phi_b\}_{\Phi=0}=\{f,\rho_a^m(p)-\rho_a^m(A)\}_{\Phi=0}\,,
\end{equation}
and using (\ref{gccF}) we obtain 
\begin{eqnarray}\label{gcc2}
\delta_f {\cal F}_{ab}&=&\{{\cal F}_{ab},f\}+
\left[\{f,\rho_a^m(p)\}+\rho_a^c(p)\,\partial_p^ m\{f,p_c\}\right]_{\Phi=0}\rho_b^n(A)\,F_{mn}\notag\\
&+&\rho_a^m(A)\left[\{f,\rho_b^n(p)\}+\rho_b^c(p)\,\partial_p^ n\{f,p_c\}\right]_{\Phi=0}F_{mn}\,.\notag
\end{eqnarray}
We thus conclude that the field strength (\ref{pfs}) transforms covariantly,
\begin{equation}
\delta_f {\cal F}_{ab}=\{{\cal F}_{ab},f\}\,,\label{gcc}
\end{equation}
if $\rho_a^m(x,p)$ satisfies the equation
\begin{equation}\label{rho}
\{f(x),\rho_a^i(x,p)\}+\rho_a^b(x,p)\,\partial_p^ i\{f(x),p_b\}=0\,, \qquad \forall f(x)\,.
\end{equation}
In local coordinates this yields
\begin{equation}\label{second}
\gamma^j_b\,\partial^b_p\, \rho_a^i+\rho_a^b\,\partial^i_p\gamma^j_b+\Theta^{jb}\,\partial_b\rho_a^i=0\,.
\end{equation} 
{Imposing the constraint~\eqref{Phi}, i.e. setting  $\rho^i_a(x,A(x))=\rho^i_a(x,p)_{\Phi=0}$ we thus recover the second master equation~\eqref{master2}. }

 {Let us now   discuss the covariant derivative. To this, let  $\psi$ be a field  which transforms upon the deformed gauge transformations according to the rule 
\be\label{gmf}
\delta_f\psi:=\{f,\psi\}.
\ee
Using the  new basis of constraints we can define the covariant derivative $\cal{D}\psi$ in the following way
\begin{equation}\label{gcd}
{\cal D}_a\psi:=\{\psi,\Phi^\prime_a\}_{\Phi=0}
\end{equation}
which explicitly yields
\be
{\cal D}_a\psi=\rho_a^m(A)\, (\gamma_a^\ell(A) \del_\ell \psi+\{A_a,\psi\}).
\ee
Notice that this result was already found in \cite{Kup33} by imposing covariance under gauge transformation, but the interpretation in terms of the new constraint was missing. By direct computation it can be checked \cite{Kup33} that it   transforms correctly, namely
\begin{equation}\label{gccd}
\delta_f\left({\cal D}_a\psi\right)=\{{\cal D}_a\psi,f\},
\end{equation}
and reproduces the commutative limit
\begin{equation}\label{w5a}
\lim_{ {\Theta}\to0}{\cal D}_a\psi=\partial_a\psi\,.
\end{equation} }

{Summarising, we see that \emph{the deformed field strength and the deformed covariant derivative are respectively defined through the Poisson brackets~\eqref{pfs} and~\eqref{gcd} of the ``rotated" constraints~\eqref{Phi1} with themselves and with the gauge parameter. The role of the matrix $\rho$ is clear as well: it performs a redefinition of  the constraints, in the extended space, obtained via the symplectic embedding.}}

As a final remark to this section, one may ask what if we start all over and replace the    transformation of gauge potentials \eqn{gtA} with an analogous  definition in terms of the newly defined constraints (\ref{Phi1}),
\begin{equation}
\delta^\prime_f A_a:=\{f,\Phi^\prime_a\}_{\Phi^\prime=0}=\rho_a^m(A)\,\delta_fA_m\,.
\end{equation}
It can be checked that the latter  does not close the desired gauge algebra, yielding
\begin{equation}
\left[\delta^\prime_f,\delta^\prime_g\right]\neq\delta^\prime_{\{f,g\}}\,.
\end{equation}
So, as matter of fact, there are two   sets of constraints,  one, given by   Eq. (\ref{Phi}), which is needed for the definition of  Poisson gauge transformations (\ref{gtA}), the other, represented by  Eq. (\ref{Phi1}), which allows for the definition of the field strength and the covariant derivative, with the desired covariance property (\ref{gcc}).

{As we have seen in this section, the key ingredients of the Poisson gauge theory, viz the deformed gauge transfomation, the deformed  field strength, and the deformed covariant derivative, are completely determined in terms of the matrices $\gamma^j_b(x,A)$ and $\rho_a^i(x,A)$. In the next section we will show how to put these objects together in order to construct the dynamical equations of motion, which exhibit the deformed version of the first pair of the Maxwell's equations. The corresponding Lagrangian formulation will allow us to obtain  \emph{new} results for the Noether identities.  

We will also consider the deformed version of the second pair of the Maxwell's equations. The original derivation, presented in~\cite{Kup33}, was actually based on a brute-force deformation of the corresponding commutative counterpart without clear connection to the symplectic embeddings.
The new derivation, instead, being more elegant, exploits the symplectic geometric construction, discussed above.  We start from the latter.
}

\section{Maxwell-Poisson equations}\label{Maxpoi}
The two Maxwell equations which correspond to gauge constraints, namely those summarised by the Bianchi identity, descend from the Jacobi identity for the modified constraints (\ref{Phi1}),
\begin{equation}
\{\Phi^\prime_a,\{\Phi^\prime_b,\Phi^\prime_c\}\}+\mbox{cycl} (abc)=0\,.
\end{equation}
   {Indeed, using {the identity (see \eqn{r2})
   \begin{equation}\label{id3}
\{F(x,p),G(x,p)\}_{\Phi=0}-\{F(x,p),G(x,A(x))\}_{\Phi=0}=\left(\partial^m_pG(x,p)\right)_{\Phi=0}\,\{F(x,p),\Phi_m\}_{\Phi=0}\,,
\end{equation}}
   and cyclic permutations one finds,
\begin{eqnarray}
&&\{\Phi^\prime_a,\{\Phi^\prime_b,\Phi^\prime_c\}\}_{\Phi=0}+\mbox{cycl} (abc)=\\
&&{\cal D}_a\left({\cal F}_{bc}\right)+\rho_a^i(A)\left[\{ \Phi_i,\rho_b^j(A)\}_{\Phi=0}\,\rho_c^k(A)+\rho_b^j(A)\,\{\Phi_i,\rho_c^k(A)\}_{\Phi=0}\right]\{\Phi_j,\Phi_k\}_{\Phi=0}+\notag\\
&&\rho_a^i(A)\,\rho_b^j(A)\,\rho_c^k(A)\left(\partial^m_p\{\Phi_j,\Phi_k\}\right)_{\Phi=0}\left\{\Phi_m,\Phi_i\right\}_{\Phi=0}+\mbox{cycl.}(abc)\,.\notag
\end{eqnarray}
By explicitly computing the Poisson brackets of constraints (see \cite{Kup33}, Sec. 5.1 for details)
 one arrives at}
\begin{equation}\label{bi}
{\cal D}_a\left({\cal F}_{bc}\right)-{\cal F}_{ad}\,\mathcal{B}_b{}^{de}\,{\cal F}_{ec}-({\cal K}_{ab}{}^e-{\cal K}_{ba}{}^e)\,{\cal F}_{ec}+\mbox{cycl}(abc)=0\,,
\end{equation}
where, 
\begin{eqnarray}\label{F}
\mathcal{B}_b{}^{de}(A)&=&\left(\rho^{-1}\right)_j^d\left(\partial^j_A\rho_b^m(A)-\partial^m_A\rho_b^j(A)\right)\left(\rho^{-1}\right)_m^e\,,\label{Lambda}\\
{\cal K}_{ab}{}^e(A)&=&\rho_a^i(A)\,\gamma^m_i(A)\,\left(\partial_m\rho_b^j(x,p)\right)_{\Phi=0}\left(\rho^{-1}\right)_j^e\,.\label{Kappa}
\end{eqnarray}
Both $\mathcal{B}$ and $\mathcal{K}$ tend to zero in the classical limit  $ {\Theta}\to 0$, so that Eq. (\ref{bi}) reproduces the correct classical  result   $\del_aF_{bc} +{\rm cycl} (abc)=0$. Moreover, the first term of the identity is gauge
covariant by construction, implying that the whole expression is gauge covariant (which could be however checked by direct inspection).

Now we turn to the gauge covariant deformation of the remaining Maxwell equations, namely those with true dynamical content, $\partial_aF^{ab}_0=0$. Taking into account (\ref{gcc}) and (\ref{gccd}) the natural candidate reads,
\begin{equation}\label{Es}
{\cal E}_{N}^b:={\cal D}_a {\cal F}^{ab} =0\,,
\end{equation}
where the subscript $N$ stands for ``natural''. This quantity transforms covariantly, $\delta_f{\cal E}_{N}^b=\{{\cal E}_{N}^b,f\}$, and reproduces the dynamical Maxwell equations in the classical  limit, ${\cal E}_{N}^b{\rightarrow}\partial_aF^{ab}_0$ for ${ {\Theta}\to0}$. 

An alternative way of obtaining  the deformed equations of motion    was proposed in \cite{Kupriyanov:2020sgx}, starting from  an action principle. Basically the idea is the following: having in hands the gauge covariant Poisson field strength (\ref{pfs}) there is a natural gauge covariant deformation of the standard Lagrangian, which reads
\begin{equation}
{\cal L}_{g}=-\frac14\,{\cal F}_{ab}\,{\cal F}^{ab}\,,\qquad\mbox{with}\qquad \delta_f{\cal L}_{g}=\{{\cal L}_{g},f\}\,.
\end{equation}
Then, introducing an appropriate measure $\mu(x)$, such that for any two Schwartz functions $f$ and $g$ the following holds
\begin{equation}\label{measure}
\int\!\dd^nx\,\mu(x)\,\{f,g\}=0\,\qquad  \Leftrightarrow \qquad \partial_l\left(\mu(x)\,\Theta^{lk}(x)\right)=0\,,
\end{equation}
one constructs the gauge invariant action,
\begin{equation}\label{ag}
S_g=\int\!\dd^nx\,\mu(x)\,{\cal L}_{g}\,,\qquad \mbox{such that,}\qquad \delta_f S_g=0\,.
\end{equation}
The corresponding Euler-Lagrange equations,
\begin{equation}\label{EL1}
{\cal E}^b_{EL}:=\frac{\delta S_g}{\delta A_c}=0\,,
\end{equation}
 are gauge covariant by construction. 
 Without going into the tedious calculations which can be found in \cite{Kup33} we write here,
\begin{eqnarray}\label{EL}
{\cal E}^b_{EL}=\rho_c^b\left[\mu\,{\cal D}_a\left({\cal F}^{ac}\right)+\frac{\mu}{2}\,{\cal F}^{cb}\,{\cal B}_b{}^{de}\,{\cal F}_{de}-\mu\,{\cal F}^{db}\,{\cal B}_b{}^{ce}\,{\cal F}_{de}+\left(\rho^{-1}\right)^c_k{\cal F}^{ab}\partial_i\left(\mu\,\rho_a^l\,\rho_b^k\,\gamma^i_l\right)\right]\,.
\end{eqnarray}

It is worth mentioning here that for some specific choices of the space-time Poisson  structure the equations of motion constructed according to (\ref{Es}) and (\ref{EL}) are equivalent. In particular, for the $su(2)$-like Poisson structure, {$\Theta^{ab}(x)=2\,\alpha\,\varepsilon^{ab}{}_c\,x^c$}, the integration measure is constant $\mu(x)=1$ and the additional terms in (\ref{EL}) vanish in such a way that,\footnote{In this particular situation the matrix $\rho_a^b(A)$ plays the role of a Lagrangian multiplier in the sense of non-Lagrangian systems \cite{GK}. The original equations of motion, ${\cal E}^a_{N}=0$, are non-Lagrangian, however the multiplication by the non-degenerate matrix $\rho_a^b(A)$ transforms them into an equivalent set of Euler-Lagrange equations, that is  (\ref{EL1}),  for the action (\ref{ag}). }
\begin{equation}\label{ELS}
{\cal E}^b_{EL}=\rho_a^b\,{\cal E}^a_{N}\,.
\end{equation}
On the other hand, for the $\kappa$-Minkowski non-commutativity the integration measure $\mu(x)$ is non-trivial \cite{Pachol:2015qia, KKV}, so the relation (\ref{ELS}) does not hold. In  Sects. \ref{cases} and  \ref{kappa} we will discuss these two cases in more detail.

The advantage of working with the action principle formalism is that in a reasonably simple way we may introduce sources and derive the corresponding conservation equations, as well as the Noether identities for the original equations of motion. Following the standard approach we introduce the current $j_a(x)$ adding the term,
\begin{equation}\label{sint}
S_{int}=-\int \dd^n x \,\mu(x) \,j^a \,A_a\,,
\end{equation}
to the action (\ref{ag}). The resulting Euler-Lagrange equations become,
\begin{equation}\label{eoms}
{\cal E}^a_{EL}=j^a\,.
\end{equation}
The interaction  term should be gauge invariant, i.e.,
\begin{eqnarray}
\delta_fS_{int}&=&-\int \dd^nx\,\mu(x) \,j^a\,\left(\gamma^l_a(A)\,\partial_lf(x)+\{A_a(x),f(x)\}\right) \\
&=&\int\dd^n x\left[\partial_l\left(\mu \,j^a\,\gamma^l_a(A)\right)+\mu\,\{A_a,j^a\}\right]f\equiv0\,.
\end{eqnarray}
Thus gauge invariance of the action implies the current conservation equation,
\begin{equation}\label{current}
\partial_l\left(\mu \,j^a\,\gamma^l_a(A)\right)+\mu\,\{A_a,j^a\}=0\,.
\end{equation}
The same logic applied to the action (\ref{ag}) results in the Noether identities for the equations of motion,
\begin{equation}\label{Noether}
\partial_l\left(\mu \,{\cal E}^a_{EL}\,\gamma^l_a(A)\right)+\mu\,\{A_a,{\cal E}^a_{EL}\}=0\,.
\end{equation}
In the classical limit $\Theta\to0$ both (\ref{current}) and (\ref{Noether}) reduce to the standard relations, $\partial_a\,j^a=0$, and $\partial_a\partial_b\,F^{ab}=0$.

The  current conservation condition \eqn{current}, together with explicit solutions of the equations of motion \eqn{eoms} deserve further investigation. In order to understand their physical meaning, a starting approach could be  to address the problem   within specific space-time models. We plan to come back to this issue  in future research. 

{Once the prescription for the construction of Poisson gauge theory is established, we are interested in the explicit solutions of the master equations~\eqref{master1} and~\eqref{master2}. In previous works we have investigated several particular examples of linear non-commutative structures, such as, the kappa-Minkowski, the $su(2)$ and the $\lambda$-Minkowski cases.  In all these cases,  the  master equations were solved explicitly. As we announced in the Introduction, the universal solution of the first master equation is known for any linear non-commutativity~\cite{KS21}. In the subsequent section we shall obtain a universal solution of the second master equation.}

\section{Linear Poisson structures}\label{Poili}
{Consider} a  linear Poisson structure 
\be
 {\Theta^{ab} =  f^{ab}_{c}\, x^c \la{linTheta}.}
\ee 
The quantities $f^{ab}_{c}$, are  structure constants, 
satisfying the Jacobi identity
\be
f^{kl}_{i} f^{ja}_{l}  + f^{j l}_{i} f^{ak}_{l} + f^{al}_{i}f^{kj}_{l} = 0. \la{Jacobi}
\ee
  In this case, one can consider  solutions  which do not depend on $x$ explicitly, so the master equations~\eqref{first} and~\eqref{second} reduce to
{\be
\gamma_{ i }^{ j } \partial^{ i }_A \gamma^{ k}_{ l} - \gamma^{ k}_{ i } \partial_A^{ i } \gamma^{ j }_{ l} 
- \gamma^{ i }_{ l}f_{i}^{jk} = 0\la{firstlin}
\ee}
and
{\be
\gamma^j_l \partial_A^l\rho_a^i + \rho_a^l \partial_A^i\gamma_l^j = 0 \la{secondlin}
\ee}
respectively.
Interestingly, for any  $\Theta$ of the form~\eqref{linTheta},   Eq.~\eqref{firstlin} has been solved  in terms of a \emph{universal} matrix function \cite{KS21}. The latter may be conveniently described 
by  introducing  the following notation
 {\be
\gamma(A) = G(\hat{A}),  \quad\quad \hat{A} \equiv A_{a}\mathfrak{e}^{a},  \la{gensolu1}
\ee}
with\footnote{Ref.~\cite{KS21} operates with the function $\chi(u) = \sqrt{\frac{u}{2}} \cot{\sqrt{\frac{u}{2}}} - 1$ and the notation $M = - \hat{A}^2 $. 
} the function $G$  given by,
 {\be
G(p) := \frac{\ii \, p}{2} + \frac{p}{2} \cot{\frac{p}{2}} = \sum_{n=0}^{\infty} (\ii p)^n B_n^{-}.   \la{Gdef}
\ee}
Here $B_n^{-}$, $n\in\mathbb{Z}_+$ stand for the Bernoulli numbers (with the index ``minus"), and $\mathfrak{e}^{a}$, $a=0,...,n-1$ are    $n\times n$ matrices so defined:
\be
[\mathfrak{e}^{a}]^{b}_c = -\ii f^{ab}_{c}. \la{commrel}
\ee
When all $\mathfrak{e}^{a}$ are linearly independent,  they are the
generators of the adjoint representation of an $n$-dimensional Lie algebra $\mathfrak{g}$, defined by the structure constants $f^{ab}_c$.\footnote{More precisely, the linear Poisson bracket \eqn{linTheta} becomes the Kirillov-Souriau-Konstant bracket which is defined on the dual of $\mathfrak{g}$, $\mathfrak{g}^*\equiv\mathbb{R}^n$. Lie algebra type Poisson brackets are also referred to as Lie-Poisson brackets, not to be confused with Poisson-Lie brackets, which appear in the semi-classical limit of quantum groups. } In particular, the commutation relation
 \be
 [\mathfrak{e}^{a},\mathfrak{e}^b] = \ii f^{ab}_{c} \mathfrak{e}^c,
 \ee 
is equivalent to the Jacobi identity~\eqref{Jacobi}.  

In order to solve  the second master equation, Eq. \eqn{secondlin}, we make the hypothesis that, similarly to the first master equation,  there exists a \emph{universal} function $F(p)$, such that the matrix function,
 {\be
\rho(A) = F(\hat{A}), \la{conjecture}
\ee
is a a solution
for any Lie-Poisson bivector $\Theta$.
%
} In order to prove the existence of $F$, we start from the $\mathfrak{su}(2)$ case, elaborating on a known solution.

\subsection{Candidate solution} \label{candidate}
For the $\mathfrak{su}(2)$-case\footnote{Let us recall that  a star product  associated with this kind of noncommutatiivity was originally introduced in \cite{hammou}.},
\be
f^{jk}_{l}  =  2\alpha\,\varepsilon^{j k}_{~~\,l},\quad \quad \varepsilon^{j k}_{~~\,l}:= \varepsilon^{j k s}\,\delta_{s l}, \la{su2}
\ee
a solution of the second master equation was found in ~\cite{Kup33}. 
 It reads
\be
\rho^i_a(A) = \delta_a^i + \alpha \,\varepsilon^{ik}_{~~a}\,A_k\, \zeta(z\,\alpha^2 ) -  \alpha^2\,(\delta_a^i \,z - A^iA_a)\,\tau  (z\,\alpha^2), \quad\quad z := A_{ i } A^{ i },\quad\quad A^{ j} := \delta^{ j\xi} A_{\xi}. \la{su2soluRho}
\ee
Let us show that this solution can be put in  the form \eqn{conjecture} for a specific function $F$.
 Both form factors in \eqn{su2soluRho}
\be
\zeta(v) := -\frac{(\sin{\sqrt{v}})^2}{v} = \sum_{k=0}^{\infty}\,\zeta_k\, v^k,\quad\quad \tau  (v) = -\frac{1}{v}\cdot\left(\frac{\sin{2\sqrt{v}}}{2\sqrt{v}} - 1\right) =\sum_{k=0}^{\infty}\,\tau_k \, v^k,
\ee
are analytic functions of the variable $v$. As for the third term in ~\eqref{su2soluRho},
introducing the projector
\be
\hat{M}_{i}^{j} :=  {\delta}^{j}_{i} - \frac{{A}^{j} {A}_{i}}{z}, \quad\quad \hat{M}^2 = \hat{M},
\ee
 we obtain\footnote{A similar computational trick has been used in~\cite{Kurkov:2021kxa}.}:
\bea
&&\!\!\!\!\!\!\!\!\!\! -\alpha^2\,(\delta_a^i \,z - A^iA_a)\,\tau  (z\,\alpha^2) =-\hat{M}^{i}_a\cdot(s\tau  (s))\big|_{s = z\alpha^2}
 =- \sum_{k=0}^{\infty} \tau_k  \, (z\alpha^2)^{k+1}\,\underbrace{\hat{M}^{i}_a}_{[\hat{M}^{k+1}]^{i}_a} \\
&&\!\!\!\!\!\!\!\!\!\!= -\sum_{k=0}^{\infty} \tau_k  \, [(z\alpha^2\hat{M})^{k+1}\,]^{i}_a  =-\Big[(s\tau  (s))\big|_{s = z\alpha^2\hat{M}}\Big]^{i}_a 
= -\Big[\bigg(\frac{p^2}{4}\,\tau  \bigg(\frac{p^2}{4}\bigg)\bigg)\bigg|_{p = \hat{A}}\bigg]^i_a
 =: T (\hat{A})_a^i, \nn
\eea
where we took into account the fact that\footnote{We remind that $\hat{A}^{k}_l = -\ii f^{jk}_l A_j $, see Eq.~\eqref{gensolu1} and Eq.~\eqref{commrel}.}
$
\frac{\hat{A}^2}{4} = z\,\alpha^2 \hat{M},
$
and the form factor $T $ is given by
{\be
T (p) := \frac{\sin{p}}{p} - 1. 
\ee}

Now we elaborate the second term of Eq.~\eqref{su2soluRho}.
For the $\mathfrak{su}(2)$ case one can easily prove, e.g. by induction, the equalities
\be
\hat{A}^i_j\,(4 \,z\, \alpha^2 )^n = \left[\hat{A}^{2n+1}\right]^i_j,\quad\quad \forall n\in\mathbb{Z}_+.
\ee 
Therefore for a generic \emph{even} analytic function 
$
Q(w) = \sum_{k=0}^{\infty} c_k \,w^{2k}
$
the following relation holds:
\be
\hat{A}^i_j \,Q(2\,\sqrt{z}\,|\alpha|) = \left[\hat{A}\,Q(\hat{A})\right]^{i}_{j}.
\ee
On chosing  $Q(w) = \zeta(w^2/4)$, we rewrite the second term of~\eqref{su2soluRho} as follows:
\be
\alpha \,\varepsilon^{ik}_{~~a}\,A_k\, \zeta(z\,\alpha^2 ) =-\frac{1}{2}\,\hat{A}^i_a\cdot\zeta(z\,\alpha^2 ) = Z(\hat{A})^i_a,
\ee
where 
{\be
U(p) = 2\ii \cdot\frac{\big(\sin{\frac{p}{2}}\big)^2}{p}.
\ee}
Finally, noticing that
\be
1+T (p) = \int_0^1\dd \beta \cos{(\beta p)},\quad\quad U(p) = \ii\int_0^1\dd \beta \sin{(\beta p)},
\ee
we get a simple answer for the undetermined function in Eq.~\eqref{conjecture}. 
 {\be
F(p)= 1+T(p) + U(p)  = 
\int_0^{1} \dd \beta\, e^{\ii \beta p}  = \frac{e^{\ii p} -1}{\ii \, p} . \la{Fdef}
\ee}
Now we demonstrate  that, in the  $\mathfrak{su}(2)$-case, there is a simple connection between $\gamma$ and $\rho$.
Calculating the inverse of the expression~\eqref{Fdef}, and comparing with Eq.~\eqref{Gdef}, one can easily establish the following relation between the functions $F$ and $G$,
\be
\frac{1}{F(p)} = \frac{\ii \, p}{e^{\ii p} -1} =  -\frac{\ii \, p}{2} + \frac{p}{2} \cot{\frac{p}{2}} = -\ii\, p + G(p),
\ee
which implies an intriguing connection~\eqref{irl} between the matrices $\gamma$ and $\rho$, announced in the Introduction.
Below we shall prove that this relation holds for an \emph{arbitrary} linear Poisson structure, which,  for $\gamma$ given by \eqn{gensolu1}, is equivalent to prove the conjecture~\eqref{conjecture}. 

\subsection{Universality of the  solution for $\rho$}
Since $\rho$ is non degenerate by hypothesis, instead of Eq.~\eqref{secondlin} we can equivalently write,  for $\rho^{-1}$,
 {\be
\gamma^k_i \partial_A^i\big[ \rho^{-1}\big]^j_l-\big[ \rho^{-1}\big]^j_i \partial_A^i\gamma_l^k=0\,.\label{Third}
\ee}
On using ~\eqref{irl}   and setting $\gamma$ equal to the universal solution~\eqref{gensolu1} we obtain 
\bea
\gamma^k_i \partial_A^i\big[ \rho^{-1}\big]^j_l-\big[ \rho^{-1}\big]^j_i \partial_A^i\gamma_l^k 
 &=& \gamma^{k}_i\partial_A^i\gamma_l^j - \gamma_i^j\partial_A^i\gamma_l^k - \ii\gamma_i^k\partial_A^i\hat{A}^j_l + \ii \hat{A}_i^j\partial_A^i\gamma_l^k \nonumber\\
 &=&   - \big[ {G}(\hat{A})\big]_i^k\, f_l^{ij} - \big[{G}(\hat{A})\big]_l^if_i^{jk} + \ii \,\hat{A}^j_i \,\partial_{A}^i \,\big[{G}(\hat{A})\big]_l^k\la{intermid}
\eea
where   Eq.~\eqref{firstlin} has been used. 
But this can be seen to be equal to zero, due to the fact that Eq.~\eqref{Gdef} defines an analytic function at a point $p=0$, together with the following result: 
\begin{proposition}
For any \emph{arbitrary}  function $S(p)$,   which is analytic at $p=0$, it holds:
\\
 {\be
\ii \,\hat{A}^j_i \,\partial_{A}^i \,\big[S(\hat{A})\big]_l^k  = \big[ S(\hat{A})\big]_i^k\, f_l^{ij} + \big[S(\hat{A})\big]_l^if_i^{jk}. \la{genid} 
\ee}
\end{proposition}
\begin{proof} Expanding $S$ in Taylor series, we see that Eq.~\eqref{genid} is equivalent to the following relations:
\be
\ii \,\hat{A}^j_i \,\partial_{A}^i \,\big[\hat{A}^n\big]_l^k  = \big[ \hat{A}^n\big]_i^k\, f_l^{ij} + \big[\hat{A}^n\big]_l^if_i^{jk},\quad\quad \forall n\in\mathbb{Z}_+ \la{genidMono} 
\ee
which can be verified  by induction. At $n = 0$ this equation becomes:
\be
 0 =  \delta_i^k \,f_l^{ij} + \delta_l^i\, f_i^{jk},
\ee 
what is obviously true, since $f_{l}^{kj} = -f_l^{jk}$. Now, assuming  ~\eqref{genidMono} to be true for $n-1$, namely that
\be
\ii \,\hat{A}^j_i \,\partial_{A}^i \,\big[\hat{A}^{n-1}\big]_{s}^k  = \big[ \hat{A}^{n-1}\big]_i^k\, f_{s}^{ij} + \big[\hat{A}^{n-1}\big]_{s}^if_i^{jk},\quad\quad n\in\mathbb{N}. \la{genidMonoNm1}
\ee
we write the the LHS of~\eqref{genidMono} at $n$ as follows
\bea
\ii \,\hat{A}^j_i \,\partial_{A}^i \,\big[\hat{A}^n\big]_l^k &=& \ii \,\hat{A}^j_i \,\partial_{A}^i \,\big(\big[\hat{A}^{n-1}\big]_s^k\cdot \hat{A} _l^s\big) \nonumber\\
&=& \ii \,\hat{A}^j_i \,\partial_{A}^i \big[\hat{A}^{n-1}\big]_s^k\cdot \hat{A}_l^s + \ii \,\hat{A}^j_i \,\big[\hat{A}^{n-1}\big]_s^k\cdot \partial_{A}^i \hat{A}_l^s \nonumber\\
&=&\big[ \hat{A}^{n-1}\big]_i^k\, f_{s}^{ij}\cdot \hat{A}_l^s + \big[\hat{A}^{n-1}\big]_{s}^if_i^{jk}\cdot \hat{A}_l^s +\hat{A}_i^j\,f_{l}^{is}\,\big[\hat{A}^{n-1}\big]_s^k, \la{interm1}
\eea
where  the assumption~\eqref{genidMonoNm1} has been used. 
The second term in the last line of Eq.~\eqref{interm1} reproduces the second term of the RHS is  ~\eqref{genidMono}:
whilst the remaining terms can be rewritten as follows
\bea
\big[ \hat{A}^{n-1}\big]_i^k\, f_{s}^{ij}\cdot \hat{A}_l^s +\hat{A}_i^j\,f_{l}^{is}\,\big[\hat{A}^{n-1}\big]_s^k &=&\big[\hat{A}^{n-1}\big]^k_s\,\big(f_i^{sj}\,\hat{A}^i_l + f^{is}_l\,\hat{A}^j_i\big) \nonumber\\
&=& \big[\hat{A}^{n-1}\big]^k_s\,\hat{A}^s_i \,f_l^{ij}+ \big[\hat{A}^{n-1}\big]^k_s\,\big(f_i^{sj}\,\hat{A}^i_l + f^{is}_l\,\hat{A}^j_i - \hat{A}^s_i f_l^{ij}\big) \nonumber\\
&=&  \big[\hat{A}^{n}\big]^k_i \,f_l^{ij} -\ii\, A_{r}\,\big[\hat{A}^{n-1}\big]^k_s\underbrace{\big(f_l^{ri}\,f_i^{sj} +f_l^{si}\,f_i^{jr} +f_l^{ji}\,f_i^{rs}\big)}_0 \nonumber\\
&=&\big[\hat{A}^{n}\big]^k_i \,f_l^{ij},
\eea
where use has been made of the relation $\hat{A}^{k}_l = -\ii f^{jk}_l A_j $. This is the the first term in the RHS of~\eqref{genidMono}, thus proving Eq. \eqn{genid} by induction.
\end{proof} 
Therefore, we can  conclude that Eq. \eqref{conjecture}, or equivalently \eqn{irl},  hold true for any linear Poisson bracket.

{We have therefore found the universal solutions for both the matrices $\gamma$ and $\rho$ in terms of  matrix-valued functions. In the following two sections we illustrate on a few nontrivial examples  how these matrix-valued functions can be calculated explicitly. }

\section{ A family of four-dimensional spaces with commutative time}\label{cases}

As an application    {of the general formulae obtained in the previous Section} we consider here a family of four dimensional models with  three-dimensional non commutativity, introduced in~\cite{Kurkov:2021kxa}. We use the Greek letters $\mu$, $\nu$, ..., and the Latin letters $a$, $b$, $c$, ...,  to denote respectively the four-dimensional and the three-dimensional (i.e. the spatial) coordinates.  
The two-parameter family of  Poisson structures to be considered reads:
 {\be
\Theta^{0\mu} = 0 = \Theta^{\mu 0}, \quad \Theta^{jk} =- \lambda\, \varepsilon^{jks}\,\check\beta_{s l} \,x^l, \la{Ps}
\ee}
where the $3\times 3$ matrix $\check{\beta}$ is given by:
 {\be
\check{\beta} := \mathrm{diag}\,\{1,\,1,\,\beta \}, \quad \beta\in\mathbb{R}.
\ee}
At $\beta = 0$ we get the Poisson structure which corresponds to  angular noncommutativity~\cite{angular1, angular2, angular3, angular4, angular5, angular6, angular7}, 
  at $\beta = 1$ the three-dimensional bivector $\Theta^{jk}$ is nothing but the Poisson structure of the $\mathfrak{su}(2)$ case~\cite{hammou},
and $\beta= -1$  corresponds to the Lie algebra  $\mathfrak{su}(1,1)$. The structure constants read
\be
 f^{\mu \nu}_{\rho}  = - \lambda\, \boldsymbol{\delta}^{\mu}_{j}\,\boldsymbol{\delta}^{\nu}_{k}\,\boldsymbol{\delta}^{l}_{\rho}\,\varepsilon^{j k s}\,\check\beta_{s l}, \la{ourf}
\ee
where 
\be
\boldsymbol{\delta}_{\mu}^{\nu} := \delta_{\mu}^{\nu} - \delta_{\mu}^0\delta^{\nu}_0,
\ee
is a projector on the three-dimensional space. The four by four matrix $\hat{A}$ is given by:
\be
\hat{A}^0_0 = \hat{A}^0_j = \hat{A}^j_0 = 0, \quad \hat{A}^k_l =\ii \lambda\varepsilon^{ksj} {\check\beta}_{sl}A_j.
\ee
On introducing the operator
\be \la{Znotations}
\big[\hat{M}_{\boldsymbol{\beta}}\big]_{\mu}^{\nu} =  {\delta}^{\nu}_{\mu} - \frac{\boldsymbol{A}_{\boldsymbol{\beta}}^{\nu} \boldsymbol{A}_{\mu}}{Z_{\boldsymbol{\beta}}},\;\;\;\;\;\;\;\; Z_{\boldsymbol{\beta}} = A_{\mu} \boldsymbol{A}_{\boldsymbol{\beta}}^{\mu} =  \beta\cdot \left(A_1\right)^2 + \beta\cdot \left(A_2\right)^2 +  \left(A_3\right)^2, 
\ee
with 
\be
\boldsymbol{A}_{\mu}  := \boldsymbol{\delta}_{\mu}^{\nu} A_{\nu}, \;\;\;\;\;\;
\boldsymbol{A}_{\boldsymbol{\beta}}^{\mu} := \hat{\beta}^{\mu \xi} A_{\xi} , \quad
\hat{\beta} := {\mathrm{diag}\{0, \beta,\,\beta,\, 1\}}.
\ee
one can easily check that 
\be
\big[\hat{A}^2\big]_{\mu}^{\nu} = \lambda^2\, Z_{\boldsymbol{\beta}}\,\big[\hat{M}_{\boldsymbol{\beta}}\big]_{\mu}^{\nu}.
\ee

The matrix $\hat{M}_{\boldsymbol{\beta}}$
is a projector, i.e. $\hat{M}_{\boldsymbol{\beta}}^2 = \hat{M}_{\boldsymbol{\beta}}$, and hence,
\be
\hat{M}_{\boldsymbol{\beta}}^n = \hat{M}_{\boldsymbol{\beta}}, \quad\forall n\in\mathbb{N}. \la{propro}
\ee
It is worth noting that at $\beta = 0$, the quantities $Z_0$ and $ \hat{M}_{\boldsymbol{0}}$ coincide with $z$ and $ \hat{M}$ introduced in Sec.~\ref{candidate}.
Using the property~\eqref{propro}, together with $\hat M \cdot \hat A= \hat A$, one can easily check that
\bea
\hat{A}^{2n} &=& \big(\lambda^2 Z_{\boldsymbol{\beta}}\big)^n \hat{M}_{\boldsymbol{\beta}},\quad \quad n\in\mathbb{N}, \nonumber\\
\hat{A}^{2n+1} &=& \big(\lambda^2 Z_{\boldsymbol{\beta}}\big)^n \hat{A},\quad \quad n\in\mathbb{N}.
\eea
Therefore, for any even function $S_{\mathrm{even}}(p)$, which is analytic at $p=0$, and which vanishes at $p=0$
\be
S_{\mathrm{even}}(\hat{A}) = S_{\mathrm{even}}\big(\lambda \sqrt{Z_{\boldsymbol{\beta}}} \,\big) \cdot\hat{M}_{\boldsymbol{\beta}}, \la{evenrule}
\ee 
and for any odd  function $S_{\mathrm{odd}}(p)$, which is analytic at $p=0$,
\be
S_{\mathrm{odd}}(\hat{A}) = \frac{S_{\mathrm{odd}}\big(\lambda \sqrt{Z_{\boldsymbol{\beta}}}\,\big)}{\lambda \sqrt{Z_{\boldsymbol{\beta}}}}\cdot \hat{A}. \la{oddrule}
\ee
Thus, for  $S_{\mathrm{odd}}(p) = \ii p/2$ and $S_{\mathrm{even}}(p) = (p/2)\cot{(p/2)} -1$, as in    Eq. ~\eqref{Gdef},
we get
\be
G(\hat{A}) = \frac{\ii}{2} \cdot\hat{A} + \mathbb{1}+
\big((\lambda \sqrt{Z_{\boldsymbol{\beta}}}/ 2)\cdot\cot{(\lambda \sqrt{Z_{\boldsymbol{\beta}}}/2 )} -1\big)\cdot\hat{M}_{\boldsymbol{\beta}},
\ee
therefore  
 {\be
\gamma_{\mu}^{\nu}(A) =  \delta_{\mu}^{\nu} 
-\frac{1}{2}\, f_{\mu}^{ \nu\xi}\cdot A_{\xi} + 
\frac{1}{ Z_{\boldsymbol{\beta}} }\cdot \big((\lambda \sqrt{Z_{\boldsymbol{\beta}}}/ 2)\cdot\cot{(\lambda \sqrt{Z_{\boldsymbol{\beta}}}/2 )} -1\big)\cdot
(
\boldsymbol{\delta}^{\nu}_{\mu}\, Z_{\boldsymbol{\beta}} - \boldsymbol{A}_{\boldsymbol{\beta}}^{\nu} \boldsymbol{A}_{\mu}
), \la{gammasolANproj}   
\ee}
where, we remind, the structure constants are defined by Eq.~\eqref{ourf}.
At $\beta = 1$ and $\lambda = -2\alpha$ the three-dimensional part $\gamma_i^j$ reproduces the result  for the $\mathfrak{su}(2)$-case, derived in~\cite{Kupriyanov:2020sgx}. The calculation, in a slightly different form has been presented in~\cite{Kurkov:2021kxa}.

Applying Eqs.~\eqref{evenrule} and~\eqref{oddrule} to $S_{\mathrm{odd}}(p) = 2\ii\,(\sin^2p/2)/p $ and $S_{\mathrm{even}}(p)  = \sin{(p)}/p - 1$, 
and using the definition~\eqref{Fdef}, we arrive at
\be
F(\hat{A}) = \mathbb{1} 
+ 2\ii\cdot\frac{\sin^2(\lambda \sqrt{Z_{\boldsymbol{\beta}}} /2)}{\lambda \sqrt{Z_{\boldsymbol{\beta}}} } \cdot \hat{A}
+ \Bigg( 
\frac{\sin{(\lambda \sqrt{Z_{\boldsymbol{\beta}}})}  }{\lambda \sqrt{Z_{\boldsymbol{\beta}}}   }- 1
\Bigg) \cdot \hat{M}_{\boldsymbol{\beta}},
\ee
therefore
 {\be
\rho_{\mu}^{\nu}(A)  = \delta_{\mu}^{\nu}
- \frac{2\,\sin^2(\lambda \sqrt{Z_{\boldsymbol{\beta}}} /2)}{\lambda^2 Z_{\boldsymbol{\beta}}    }  \, f_{\mu}^{ \nu\xi}\, A_{\xi} 
+\frac{1}{ Z_{\boldsymbol{\beta}} }\cdot \Bigg( 
\frac{\sin{(\lambda \sqrt{Z_{\boldsymbol{\beta}}})}  }{\lambda \sqrt{Z_{\boldsymbol{\beta}}}   } - 1
\Bigg) \cdot
(
\boldsymbol{\delta}^{\nu}_{\mu}\, Z_{\boldsymbol{\beta}} - \boldsymbol{A}_{\boldsymbol{\beta}}^{\nu} \boldsymbol{A}_{\mu}
). \la{rhosolANproj}
\ee}
At $\beta = 1$ and $\lambda = -2\alpha$ the three-dimensional part $\rho_i^j$ reproduces the known formula~\eqref{su2soluRho} for the $\mathfrak{su}(2)$-case, derived in~\cite{Kup33}.

\subsection{Comments on the presence of commutative coordinates}\label{comments}
The previous example has a peculiar feature - the presence of the fourth commutative coordinate, which we associate with time, whilst the noncommutativity is essentially three-dimensional. 
The $4\times 4$ matrices $\rho_{\mu}^{\nu}$ and $\gamma_{\mu}^{\nu}$ are given by a  {simple} generalisation of the corresponding three-dimensional results, since
\be
\gamma^{\mu 0}  =  \gamma^{ 0\mu}= \delta^{\mu {0}},\quad \rho^{\mu 0} = \rho^{0\mu} = \delta^{\mu{0}}, \la{gr4}
\ee
whilst the $3\times 3$  $\gamma^{i}_j$ and $\rho_{j}^i$ solve the three-dimensional master equations.
On the other hand, it has been shown in~\cite{Kurkov:2021kxa}, that the deformed field strength $\mathcal{F}_{\mu\nu}$ exhibits a highly nontrivial behaviour in the four-dimensional case, namely it is not a simple generalisation of  the three dimensional one.  In particular the components $\mathcal{F}_{0 j}$ are nonlinear in the gauge potential 
$A$, even in the simplest case of a spatially homogeneous situation, i.e. when $A$  dose not depend on the spatial coordinates $x^{j}$.

 {Below we demonstrate that the simple behaviour~\eqref{gr4} of the universal solution, indeed  implies the nontrivial behaviour of $\mathcal{F}_{0j}$, mentioned above. }  
We have
{\be
\mathcal{F}_{ab}= (\rho_a^m \rho_b^n-\rho_a^n \rho_b^m)\gamma_m^l \del_l A_m + \frac{1}{2}(\rho_a^m \rho_b^n-\rho_a^n \rho_b^m)\{A_m, A_n\}
\ee
therefore, for $A$ independent on spatial coordinates, we find
\bea
\mathcal{F}_{0j}&=& (\delta_0^m \rho_j^n-\delta_0^n \rho_j^m)\gamma_m^l \del_l A_m + \frac{1}{2}(\delta_0^m \rho_j^n-\delta_0^n \rho_j^m)\{A_m, A_n\}\nn\\
&=&\delta_0^l \rho^n_j\del_l A_n- \gamma_j^l\del_l A_0 + \{A_0,A_j\}=\rho^n_j\del_0A_n,
\eea
}
namely, the electrical component of the field strength is non-trivially modified. 
Thus, by adding commutative coordinates, it is certainly true that  $\gamma$ and $\rho$ are given by  a trivial generalization of the corresponding lower dimensional results. But this is definitely an {intermediate} stage. The field strength,  which is the true dynamical object, gets instead non-trivial contributions, already in the simple hypothesis of spatial homogeneity. These conclusions go beyond the example and in general will  apply to  any non-commutative spacetime with commutative directions. 

\section{$\kappa$-Minkowski space-time}\label{kappa} 
Another important application of the results of section \ref{Poili} is  the $\kappa$-Minkowski case in $N$ dimensions \cite{kappa1,kappa2,kappa3,kappa4,kappa5,kappa6},  where the Poisson bivector is given by 
\be
\Theta^{ij} = 2(\omega^ix^j - \omega^j x^i),
\ee
where $\omega^i$, $i=1,...,N$ are deformation parameters.
Substituting the outcoming structure constants
\be
f^{aj}_k  = 2(\omega^a\delta^j_k - \omega^j\delta^a_k),
\ee
and taking into account  Eq.~\eqref{gensolu1} and Eq.~\eqref{commrel}, we obtain:
\be
\hat{A}_k^j = -\ii f_{k}^{aj}A_a = -2\ii \,(\omega\cdot A)\, \hat{P}^j_k,
\ee
where, by definition,
\be
\omega\cdot A = \omega^jA_j,
\ee
and
\be
\hat{P}^j_k = \delta^j_k - (\omega\cdot A)^{-1} \cdot\omega^j A_k.
\ee
The matrix $\hat{P}$ is a projector, i.e. $\hat{P}^2 = \hat{P}$.
Therefore, using Eq.~\eqref{gensolu1} and Eq.~\eqref{conjecture}, we can immediately calculate the matrices $\gamma$ and $\rho$.
For this purpose, we represent the functions $G$ and $F$, defined by Eq.~\eqref{Gdef} and Eq.~\eqref{Fdef} as follows
\bea
G(p) = 1 + \tilde{G}(p), \nonumber\\
F(p) = 1+ \tilde{F}(p),
\eea
where the functions $\tilde{G}(p)$ and $\tilde{F}(p)$ are analytic and vanishing at  $p=0$. For any such function, say $S$, 
the following matrix identity can be easily checked:
\be
S(\lambda\hat{P}) = S(\lambda)\hat{P},
\ee
with $\lambda$ any complex number. Therefore,
 {\be
\gamma^{i}_j (A)= \delta^i_j + \tilde{G}\big(-2\ii\,(\omega\cdot A)\big)\hat{P}^i_j = \delta^{i}_j + \big((\omega\cdot A) +  (\omega\cdot A) \coth{(\omega\cdot A)} - 1\big)\hat{P}^i_j 
 \ee}
and
 {\be
\rho_j^i(A)  = \delta_j^i + \tilde{F}(-2\ii\,(\omega\cdot A))\hat{P}^i_j = \delta^i_j + \frac{1}{2}\,(\omega\cdot A)^{-1}\big(e^{2(\omega\cdot A)}-1-2(\omega\cdot A)\big)\hat{P}^i_j.
\ee}
{Restoring the original notations, we get our final expressions for the matrices $\gamma$ and $\rho$ in the form:
\bea
\gamma^{i}_j (A) &=&  (\omega\cdot A) \left[1+   \coth{(\omega\cdot A)} \right]\delta^i_j  +\frac{1-(\omega\cdot A)-(\omega\cdot A)\coth{(\omega\cdot A)}}{\omega\cdot A}\,\omega^i A_j,\nonumber\\
 \rho_j^i(A)   &=& \frac{e^{2(\omega\cdot A)}-1}{2(\omega\cdot A)}\, \,\delta^i_j + \frac{1+2(\omega\cdot A)-e^{2(\omega\cdot A)}}{2(\omega\cdot A)^2}\,\,\omega^i A_j. \la{gammarhokappa}
 \eea
These results are different from the ones obtained in~\cite{Kup33,KKV}. In the next section we address the non-uniqueness of the construction. In particular, we will show that the new and the old expressions for $\gamma$ and $\rho$ are related via the Seiberg-Witten map.  
}

\section{Arbitrariness of the solutions and the Seiberg-Witten map}\label{seiwit}
{It has already been noticed that solutions of the master equations~\eqref{master1} and~\eqref{master2}, defining the Poisson gauge field theoretical model, are not unique.  It is easy to see that for any invertible field redefinition 
\be
A\to\tilde A(A), \la{redef}
\ee
 the quantities 
\be
\tilde{\gamma}^{i}_j (\tilde{A}) = \Bigg(\gamma_k^i(A)\cdot \frac{\partial \tilde{A}_j}{\partial A_k}\Bigg)\Bigg|_{A = A(\tilde{A})}, \la{gaSW}
\ee
and
\be
\tilde{\rho}_a^i(\tilde{A})  = \Bigg(\frac{\partial A_s}{\partial \tilde{A}_i}\cdot \rho_a^s(A)\Bigg)\Bigg|_{A = A(\tilde{A})}, \la{rhoSW}
\ee
are again  solutions of the master equations. Therefore, 
one may construct another gauge transformation corresponding to  $\tilde\gamma^i_a$,
\begin{equation}\label{fr1}
\tilde\delta_f\tilde A_a=\tilde\gamma^i_a(\tilde A)\,\partial_if+   {\{\tilde A_a,f\}}\,,
\end{equation}
which will close the same gauge algebra (\ref{ga}). Upon the field redefinition~\eqref{redef} the gauge orbits of the original fields and the gauge orbits of the new fields are mapped onto each other. Indeed, the Seiberg-Witten condition~\cite{Seiberg:1999vs},
\begin{equation}
\tilde A\left(A+\delta_ fA\right)=\tilde A(A)+\tilde\delta_f\tilde A(A)\,,
\end{equation}
is trivially satisfied up to the linear order in $f$.  Therefore the invertible field redefinitions are nothing but the Seiberg-Witten maps. 

 As we have seen above,  Poisson gauge models are actually based on  symplectic embeddings. The ambiguity, discussed above, corresponds to the freedom in choosing different symplectic embeddings for the Poisson manifolds~\cite{KS21,Kup33}. Finally, we recall that, in the equivalent approach in terms of L$_\infty$ bootstrap, it was shown in \cite{BBKT}  that  Seiberg-Witten maps for the Poisson gauge algebra under analysis  correspond in that picture to L$_\infty$-quasi-isomorphisms which describe the arbitrariness in the definition of the related L$_\infty$ algebra.
In what follows, we derive the Seiberg-Witten map for some noteworthy cases with linear non-commutativity.
}

\subsection{{The  $\kappa$-Minkowski case}}
{ One can check by a straightforward calculation that the solutions of the master equations for the $\kappa$-Minkowski non-commutativity, presented in~\cite{Kup33,KKV}, viz  
\begin{equation}
\tilde\gamma^k_a(A)= \left[\sqrt{1+(\omega\cdot A)^2}+(\omega\cdot A)\right]\delta^k_a -\omega^k\,A_a\,.\label{gammakappa2}
\end{equation}
and
\be
\tilde\rho^k_a(A) =\left[\sqrt{1+(\omega\cdot A)^2}+(\omega\cdot A)\right]\,\delta_a^k -\frac{\sqrt{1+(\omega\cdot A)^2}+(\omega\cdot A)}{\sqrt{1+(\omega\cdot A)^2}}\,\omega^k A_a,
\ee
can be obtained from the new ones~\eqref{gammarhokappa},  using the formulae~\eqref{gaSW} and~\eqref{rhoSW}. The new fields are defined in the following way
\begin{equation}
\tilde A_a=\frac{\sinh (\omega\cdot A)}{\omega\cdot A} \,A_a\,. \la{AtransKappa}
\end{equation}
Such a Seiberg-Witten map is invertible, and the inverse transformation reads: 
\begin{eqnarray}
 A_a=\frac{\mathrm{arcsinh}\,(\omega\cdot \tilde{A})}{\omega\cdot \tilde{A}} \,\tilde A_a\,.
\end{eqnarray}
}

{\subsection{The $\mathfrak{su}(2)$ case}
As we mentioned above, setting $\boldsymbol{\beta} = 1$ and $\lambda = -2\alpha$ in the three-dimensional component  of Eq.~\eqref{gammasolANproj} and Eq.~\eqref{rhosolANproj}, we get the ``standard" solution\footnote{Applying the universal formulas~\eqref{gensolu1} and~\eqref{conjecture}, \eqref{Fdef}, one gets exactly these matrices $\gamma$ and $\rho$.
}  of the master equations~\eqref{first} and~\eqref{second} for the $\mathfrak{su}(2)$-case: 
\bea
\gamma_{k}^{i}(A) &=&  \delta_{k}^{i} 
- \alpha\, \varepsilon^{il}_{~~k}A_l+ 
\frac{1}{ Z_{1} }\big(\alpha \sqrt{Z_{1}}\,\cot{(\alpha \sqrt{Z_{1}} )} -1\big)
(
{\delta}^{i}_{k}\, Z_{1} - A^{i} A_{k}
) \nonumber\\
\rho^i_a(A) &=& \delta^i_a - \alpha\, \varepsilon^{ik}_{~~a} \,A_k\, \frac{1}{Z_1}\sin^2{ \sqrt{Z_1} }
+\frac{1}{ Z_{1} }\Bigg( 
\frac{\sin{(2\,\alpha \sqrt{Z_{1} })}  }{2\,\alpha \sqrt{Z_{1 }  }   } - 1
\Bigg) 
(
{\delta}^{i}_{a}\, Z_{1} - A^{i}A_{a}
) 
\eea
which have been previously presented in~\cite{Kup33} and~\cite{Kupriyanov:2020sgx}. We remind that the quantity 
$Z_1$ is given by
\be
Z_{1} :=  \left(A_1\right)^2 +  \left(A_2\right)^2 +  \left(A_3\right)^2,
\ee
see Eq.~\eqref{Znotations} at $\boldsymbol{\beta}= 1$.

An alternative solution  of the first master equation~\eqref{first} can be obtained using the results of~\cite{Gubitosi:2021itz}
in the following form:
\be
\tilde{\gamma}^i_k(A) = \delta_k^i  - \alpha\, \varepsilon^{il}_{~~k}A_l  + \alpha^2 A^i A_k.
\ee
It is easily checked that the ``standard" and the new solution are related via the field redefinition\footnote{Note that in \cite{Smilga} this map was interpreted as a coordinate transformation relating two different representation of a sphere $S^3$.}
\be
 \tilde{A}_i(A)  := A_i \cdot \frac{\tan{\big(\alpha\sqrt{Z_1}\big)}}{\alpha\sqrt{Z_1}} \la{AtransSU2}
\ee
according to the relation~\eqref{gaSW}.
The inverse field redefinition reads:
\be
 A_j(\tilde{A}) = \tilde{A}\cdot\frac{\arctan{\big(\alpha\sqrt{\tilde{Z}_1}\big)}}{ \alpha\sqrt{\tilde{Z}_1}},
\ee
where we set by definition
\be
\tilde{Z}_{1} :=  \left(\tilde{A}_1\right)^2 +  \left(\tilde{A}_2\right)^2 +  \left(\tilde{A}_3\right)^2.
\ee
Using these transformations {together with the relation~\eqref{rhoSW}}, one can easily construct the new solution for the second master equation~\eqref{second}
\be
\tilde{\rho}^i_a(A) = \frac{1}{1+ \alpha^2\,Z_1}\cdot \big(\delta_a^i - \alpha\,\varepsilon^{ik~}_{~~a}A_k\big).   \la{rhoNewSU2}
\ee

\subsection{The $\lambda$-Minkowski case} 
For the sake of simplicity we restrict ourselves to the three-dimensional $\lambda$-Minkowski case. Indeed, since the fourth coordinate is commutative, the 
four-dimensional generalisation of $\gamma$ and $\rho$ is trivial, see the discussion in Sect. \ref{comments}.  
Setting $\boldsymbol{\beta} = 0$ in the three-dimensional parts of Eq.~\eqref{gammasolANproj} and Eq.~\eqref{rhosolANproj}, we arrive at the following ``standard" solutions of the master equations:
\be
 \gamma(A) = \left( \begin {array}{ccc} \frac{A_{{3}}\lambda}{2}\cot \left( \frac{A_{{3}}\lambda}{2} \right)  & -\frac{A_{{3}}\lambda}{2}  & 0\\ 
\noalign{\medskip}\frac{A_{{3}}\lambda}{2}                  &         \frac{A_{{3}}\lambda}{2}\cot \left( \frac{A_{{3}}\lambda}{2} \right)      &0\\ 
\noalign{\medskip} \frac {-\cot \left( \frac{A_{{3}}\lambda}{2}  \right) \lambda\,A_{{1}}-\lambda\,A_{{2}} +\frac{2\,A_{{1}}}{A_{{3}}} }{2 }
&  \frac {-\cot \left( \frac{A_{{3}}\lambda}{2}  \right) \lambda\,A_{{2}}+\lambda\,A_{{1}} +\frac{2\,A_{{2}}}{A_{{3}}} }{  2 }     &1
\end {array} \right), 
\ee
and
\be
\rho(A) = \left( \begin {array}{ccc} \frac {\sin \left( A_{{3}}\lambda\right) }{A_{{3}}\lambda}  & -{\frac { 2\left( \sin \left( \frac{A_{{3}}\lambda }{2}\right)  \right) ^{2}}{ A_{{3}}\lambda}} & 0\\ 
\noalign{\medskip} { {\frac { 2\left( \sin \left( \frac{A_{{3}}\lambda }{2}\right)  \right) ^{2}}{ A_{{3}}\lambda}}  } & \frac {\sin \left( A_{{3}}\lambda\right) }{A_{{3}}\lambda} & 0\\ 
\noalign{\medskip}{ \frac{-\frac{2A_{{2}}  \left( \sin \left( \frac{A_{{3}}\lambda}{2} \right)  \right) ^{2} }{A_{{3}}\lambda}+ \left( 1- \frac {\sin \left( A_{{3}}\lambda\right) }{A_{{3}}\lambda}  \right) A_{{1}} }{{A_{{3}} }} }&
\frac{\frac{2A_{{1}} \left( \sin \left( \frac{A_{{3}}\lambda }{2}\right)  \right) ^{2} }{A_{{3}}\lambda}+ \left( 1- \frac {\sin \left( A_{{3}}\lambda\right) }{A_{{3}}\lambda} \right) A_{{2}}}{{A_{{3}}}}&
1\end {array} \right).
\ee
The results of ~\cite{Gubitosi:2021itz} suggest a new (much simpler!) solution for the first master equation, that is:
\be
\tilde{\gamma}(A) = \left( \begin {array}{ccc} 1&0&0\\ \noalign{\medskip}0&1&0
\\ \noalign{\medskip}-\theta\,A_{{2}}&\theta\,A_{{1}}&1\end {array}
 \right).   \la{gammaNewLambda}
\ee
One can show  that the two gamma matrices  are related through the relation ~\eqref{gaSW}, where the new fields are defined as follows:
\bea
\tilde{A}_1 &=& \frac {\sin \left( A_{{3}}\lambda\right) }{A_{{3}}\lambda}\, A_1
 - \frac{2  \sin^2 \left( \frac{A_{{3}}\lambda }{2}\right)   }{A_{{3}}\lambda}\, A_2, \nonumber\\
 \tilde{A}_2 &=& \frac {\sin \left( A_{{3}}\lambda\right) }{A_{{3}}\lambda}\, A_2
 + \frac{2  \sin^2 \left( \frac{A_{{3}}\lambda }{2}\right) }{A_{{3}}\lambda}\, A_1, \nonumber\\
 \tilde{A}_3 &=& A_3.
\eea
The inverse transformation reads:
\bea
A_1 &=& \frac{\tilde{A}_{{3}}\lambda}{2}\cot \left( \frac{\tilde{A}_{{3}}\lambda}{2} \right)\,\tilde{A}_1 + \frac{\tilde{A}_{{3}}\lambda}{2}\,\tilde{A}_2 , \nonumber\\
A_2 &=& \frac{\tilde{A}_{{3}}\lambda}{2}\cot \left( \frac{\tilde{A}_{{3}}\lambda}{2} \right)\,\tilde{A}_2 - \frac{\tilde{A}_{{3}}\lambda}{2}\, \tilde{A}_1 , \nonumber\\
A_3 &=& \tilde{A}_3. \la{AtransLambda}
\eea
This Seiberg-Witten map   allows to construct a new solution of the second master equation, which corresponds to $\tilde\gamma$, via the relation~\eqref{rhoSW}:
\be \la{rhoNewLambda}
\tilde{\rho}(A) = \left( \begin {array}{ccc} \cos \left( A_{{3}}\lambda \right) &-\sin
 \left( A_{{3}}\lambda \right) &0\\ \noalign{\medskip}\sin \left( A_{{
3}}\lambda \right) &\cos \left( A_{{3}}\lambda \right) &0
\\ \noalign{\medskip}0&0&1\end {array} \right). 
\ee 
Though these new solutions of the master equations do not obey the intriguing relation \eqref{irl}, they satisfy another interesting identity:
\be
\tilde{\rho}\,\tilde{\gamma} + \mathbb{1} = \tilde{\rho} + \tilde{\gamma}.
\ee
}

\section{Conclusions}
{As we explained in the review sections of this paper,  Poisson gauge modelss are completely determined by the matrices $\gamma$ and $\rho$, which solve the two master equations.
The new results of the present article are the following.
\begin{itemize}
\item{We obtained the universal solution~\eqref{irl} of the second master equation, which is valid for any Poisson bivector,  linear in coordinates. This solution for $\rho$, together with the previously known universal solution~\eqref{gensolu1} for $\gamma$, allows to build the Poisson gauge theory completely. This is the most important result of our paper.}
\item{The expressions for $\rho$ and $\gamma$, which come out from the universal solutions, do not in general coincide with the expressions that have been obtained previously in the literature for the same noncommutativity. However, the corresponding Poisson gauge theories are connected with each other through Seiberg-Witten maps. We have constructed these maps explicitly, see Eq.~\eqref{AtransKappa}, Eq.~\eqref{AtransSU2} and~\eqref{AtransLambda}. }
\item{Using the Seiberg-Witten maps,  we obtained the new simple expressions~\eqref{rhoNewSU2} and~\eqref{rhoNewLambda} for the matrix $\rho$ in the $su(2)$ and in the $\lambda$-Minkowski cases respectively. 
The matrices~\eqref{gammaNewLambda}  and~\eqref{rhoNewLambda}, related to the $\lambda$-Minkowski noncommutativity,  provide the simplest nontrivial solutions of the master equations, which have been found so far. }
\item{Another set  of interesting results concerns the geometric interpretation of  Poisson gauge models in terms of  symplectic embeddings and  constrains in the extended space. In Eq.~\eqref{pfs} and~\eqref{gcd} we expressed the deformed field strength and the deformed gauge covariant derivatives through Poisson brackets of the ``rotated" constrains in the extended space, obtained via the symplectic embedding. Moreover, we clarified the role of the matrix $\rho$ in this geometric approach: it performs non-linear transformations of  the constrains in the extended space, see Eq.~\eqref{Phi1}.  These results, accompanied by the previously known symplectic geometric interpretation of the matrix $\gamma$ and of the deformed gauge transformation, complete the geometric interpretation of  Poisson gauge theory.}
\item{Finally, we derived the deformed Noether identities~\eqref{Noether}. }
\end{itemize} }
\appendix

\section{Useful formulae}\label{AppB}
Let us first calculate,
\begin{eqnarray}\label{r1}
&&\{f(x),\{g(x),p_a\}\}_{\Phi=0}-\{f(x),\{g(x),p_a\}_{\Phi=0}\}_{\Phi=0}~~~~~~~\\
&&~~~~~~~~~~~=\{f(x),\gamma^i_a(x,p)\,\partial_ig(x)\}_{\Phi=0}-\{f(x),\gamma^i_a(x,A(x))\,\partial_ig(x)\}_{\Phi=0}\notag\\
&&~~~~~~~~~~~=\;\;\left(\partial^b_p\{g(x),p_a\}\right)_{\Phi=0}\left(\{f(x),p_b\}_{\Phi=0}-\{f(x),A_b(x)\}_{\Phi=0}\right)\notag\\
&&~~~~~~~~~~~=\left(\partial^b_p\{g(x),p_a\}\right)_{\Phi=0}\{f(x),\Phi_b\}_{\Phi=0}\,.\notag
\end{eqnarray}
Observe that the Poisson bracket of two functions  of coordinates only does not depend on $p$-variables, so, that
\be\label{r2}
\{f(x),g(x)\}_{\Phi=0}=\{f(x),g(x)\}\,.
\ee
The combination of (\ref{r1}) and (\ref{r2}) implies,
\begin{equation}\label{r3}
\left\{f(x),\{g(x),\Phi_a\}\right\}_{\Phi=0}-\left\{f(x),\{g(x),\Phi_a\}_{\Phi=0}\right\}_{\Phi=0}=\left(\partial^d_p\{g(x),\Phi_a\}\right)_{\Phi=0}\{f(x),\Phi_d\}_{\Phi=0}\,.
\end{equation}

One may also check that,
\be\label{r5}
\{\{f(x),\Phi_b\},\Phi_a\}_{\Phi=0}-\{\{f(x),\Phi_b)\}_{\Phi=0},\Phi_a\}_{\Phi=0}=
\left(\partial^d_p\{f(x),\Phi_b)\}\right)_{\Phi=0}\{\Phi_d,\Phi_a\}_{\Phi=0}
\ee
and,
\be\label{r6}
\{\{\Phi_c,\Phi_b\},\Phi_a\}_{\Phi=0}-\{\{\Phi_c,\Phi_b)\}_{\Phi=0},\Phi_a\}_{\Phi=0}=
\left(\partial^d_p\{\Phi_c,\Phi_b)\}\right)_{\Phi=0}\{\Phi_d,\Phi_a\}_{\Phi=0}.
\ee

\bigskip

\noindent {\bf Proof of Eq. \eqn{gccF}}

\noindent We compute 
 \begin{eqnarray}
\delta_f F_{ab}=\left(\partial^c_p\{\Phi_a,\Phi_b\}\right)_{\Phi=0}\{f,\Phi_c\}_{\Phi=0}-\{\{f,\Phi_a\}_{\Phi=0},\Phi_b\}_{\Phi=0}-\{\Phi_a,\{f,\Phi_b\}_{\Phi=0}\}_{\Phi=0}
\end{eqnarray}
Using Eqs. \eqn{r2}-\eqn{r6}, we rewrite the RHS as,
\begin{eqnarray}
RHS&=&\left(\partial^c_p\{\Phi_a,\Phi_b\}\right)_{\Phi=0}\{f,\Phi_c\}_{\Phi=0}-\{\{f,\Phi_a\},\Phi_b\}_{\Phi=0}-\{\Phi_a,\{f,\Phi_b\}\}_{\Phi=0}\notag\\
&&\left(\partial^c_p\{f,\Phi_a\}\right)_{\Phi=0}\{\Phi_c,\Phi_b\}_{\Phi=0}-\left(\partial^c_p\{f,\Phi_b\}\right)_{\Phi=0}\{\Phi_c,\Phi_a\}_{\Phi=0}\,.
\end{eqnarray}
Using in the first line the Jacobi identity and  Eqs. \eqn{r2}-\eqn{r6} one more time, we end up with
\be
RHS=\{\{\Phi_a,\Phi_b\}_{\Phi=0},f\}_{\Phi=0}+\left(\partial^c_p\{f,\Phi_a\}\right)_{\Phi=0}\{\Phi_c,\Phi_b\}_{\Phi=0}-\left(\partial^c_p\{f,\Phi_b\}\right)_{\Phi=0}\{\Phi_c,\Phi_a\}_{\Phi=0}\, \nonumber
\ee
We therefore have proved that,
\begin{equation}\label{interm}
\delta_f F_{ab}=\{F_{ab},f\}_{\Phi=0}+\left(\partial^c_p\{f,\Phi_a\}\right)_{\Phi=0}F_{cb}-\left(\partial^c_p\{f,\Phi_b\}\right)_{\Phi=0}F_{ca}\,.
\end{equation}
Now remember that
\begin{equation}
\partial^c_p\{f,\Phi_a\}=\partial^c_p\{f,p_a\}\,,
\end{equation}
since $\{f(x),A_a(x)\}$ does not depend on $p$. Also the fact that $F_{ab}$ is a function of $x$ only implies that the Poisson bracket of two functions of $x$, $\{F_{ab},f\}$ does not depend on $p$-variables, so $\{F_{ab},f\}_{\Phi=0}=\{F_{ab},f\}$. 
By substituting this result in eq. \eqn{interm} we have proved \eqn{gccF}.

\end{document}